\documentclass{llncs}
\usepackage{amsmath,amssymb,amsfonts}
\usepackage{stmaryrd}
\usepackage{graphicx}
\usepackage{appendix}
\usepackage{epsf}
\usepackage{epsfig}
\usepackage{epstopdf}
\usepackage{xspace}
\usepackage{booktabs}
\usepackage{latexsym}
\usepackage{lipsum}
\usepackage[ruled,vlined]{algorithm2e}
\usepackage[linkbordercolor={0 0 1}]{hyperref}
\usepackage{nameref}

\newcommand{\NP}{\ensuremath{\sf{NP}}\xspace}
\renewcommand{\L}{\ensuremath{\sf{L}}\xspace}
\newcommand{\AC}{\ensuremath{\sf{AC}}\xspace}

\newcommand{\TC}{\ensuremath{\sf{TC}}\xspace}

\newcommand{\PL}{\ensuremath{\sf{Para}}\xspace}
\newcommand{\Pb}{\ensuremath{\sf{P}}\xspace}

\newcommand{\coAM}{\ensuremath{\sf{co}}\ensuremath{\sf{AM}}\xspace}

\newcommand{\fpt}{\ensuremath{\sf{FPT}}\xspace}

\newcommand{\vc}{\textsc{Vertex Cover}}
\newcommand{\GI}{\ensuremath{\sf{GI}}\xspace}

\title{On the Parallel Parameterized Complexity of the Graph Isomorphism Problem}
\author{Bireswar Das, Murali Krishna Enduri\thanks{Supported by Tata 
Consultancy Services (TCS) research fellowship} and I. Vinod Reddy}
\institute{IIT Gandhinagar, India \\
\email{\{bireswar,endurimuralikrishna,reddy\_vinod\}@iitgn.ac.in}
}

\begin{document}
\date{}
\maketitle
\begin{abstract}
%The graph isomorphism problem (\GI{}) is to decide whether two given input graphs are isomorphic.

In this paper, we study the parallel and the space complexity of the graph isomorphism problem (\GI{}) for  several parameterizations. 

Let $\mathcal{H}=\{H_1,H_2,\cdots,H_l\}$ be a finite set of graphs where $|V(H_i)|\leq d$ for all $i$ and for some constant $d$. Let $\mathcal{G}$ be an $\mathcal{H}$-free graph class i.e., none of the graphs $G\in \mathcal{G}$ contain any $H \in \mathcal{H}$ as an induced subgraph.  We show that \GI{} parameterized by vertex deletion distance to $\mathcal{G}$ is in a parameterized version of $\AC^1$, denoted $\PL$-$\AC^1$, provided the colored graph isomorphism problem for graphs in $\mathcal{G}$ is in $\AC^1$. From this, we deduce that \GI{} parameterized by the vertex deletion distance to cographs is in $\PL$-$\AC^1$. 

The parallel parameterized complexity of \GI{} parameterized by the size of a feedback vertex set remains an open problem. Towards this direction we show that the graph isomorphism problem is in $\PL$-$\TC^0$ when parameterized by vertex cover or by twin-cover.
%is  in   $\PL$-$\TC^0$. We also prove that the graph isomorphism problem  parameterized by twin-cover is  in   $\PL$-$\TC^0$.

Let $\mathcal{G}'$ be a graph class such that recognizing graphs from $\mathcal{G}'$ and the colored version of \GI{} for $\mathcal{G}'$ is in logspace ($\L$). We show that \GI{} for bounded vertex deletion distance to $\mathcal{G}'$ is in $\L$. From this, we obtain logspace algorithms for \GI{} for graphs with bounded vertex deletion distance to interval graphs and graphs with bounded vertex deletion distance to cographs.

 \end{abstract}

\section{Introduction}\label{intro}
Two graphs $G=(V_g,E_g)$ and $H=(V_h,E_h)$ are said to be \emph{isomorphic} if there is a bijection $f:V_g\rightarrow V_h$ such that for all pairs $\{u,v\} \in {V_g \choose 2}$,  $\{u,v\}\in E_g$ if and only if $\{f(u),f(v)\}\in E_h$. Given a pair of graphs as input the problem of deciding if the two graphs are isomorphic is known as the \emph{graph isomorphism problem} ($\GI$). 
Whether this problem has a polynomial-time algorithm is one of the outstanding open problem in the field of algorithms and complexity theory. 
%The graph isomorphism problem is not known to be in \Pb{}. 
It is in $\NP$ but very unlikely to be $\NP$-complete as it is in $\NP \cap \coAM $~\cite{Boppana1987}. 
%The problem is not even known to be hard for $\Pb$. 
%The best known algorithm for $\GI$ runs in time $2^{O(\sqrt{n\log n})}$ \cite{babai1981,Zem1982}. 
Recently Babai~\cite{babai2015graph} designed a quasi-polynomial time algorithm for \GI{}  improving the best previously known runtime  $2^{O(\sqrt{n\log n})}$~\cite{babai1981,Zem1982}.
However, efficient algorithms for \GI{} have been discovered for various restricted classes of graphs e.g., planar graphs~\cite{hopcroft1974linear}, bounded degree graphs \cite{luks1982isomorphism}, bounded genus graphs~\cite{miller1980isomorphism}, bounded tree-width graphs~\cite{bodlaender1990polynomial} etc.

For restricted classes of graphs the complexity of $\GI$ has been studied more carefully and finer complexity classifications within \Pb{} have been done. 
% The graph isomorphism problem for restricted classes of graphs has been classified in small complexity classes with in \Pb{}. 
Lindell~\cite{lindell1992logspace} gave a deterministic logspace algorithm for isomorphism of trees. In the recent past, there have been many logspace algorithms for $\GI$ for restricted classes of graphs e.g., $K_{3,3}$ or $K_5$ minor free graphs \cite{datta2009graph}, planar graphs \cite{datta2009planar}, bounded tree-depth graphs \cite{das2015logspace}, bounded tree-width graphs \cite{elberfeld2016canonizing} etc.
%Helly circular-arc graphs \cite{kobler2013helly} 
On the other hand parallel isomorphism algorithms have been designed for graphs with bounded eigenvalue multiplicity ~\cite{babai1986vegas}, bounded color class graphs~\cite{luks1986parallel} etc.
%The space complexity of \GI{} problem for bounded clique-width graphs remains an open problem. In this direction we give a logspace algorithm to \GI{} for subclasses of bounded clique-width graphs. 
%We are giving logspace algorithms to \GI{} for class of graphs containing cliques of arbitrary size which are unbounded tree-width. 

The graph isomorphism problem has been studied in the parameterized framework 
%and $\fpt$ algorithms have been designed 
for several graph classes with parameters such as  the tree-depth~\cite{bouland2012tractable}, the tree-distance width~\cite{yamazaki1999isomorphism}, the connected path distance width~\cite{otachi2012isomorphism} and recently the tree-width which corresponds to a  much larger class ~\cite{lokshtanov2017fixed}.  A more detailed list of $\fpt$ algorithms for $\GI$ in parameterized setting can be found in~\cite{Bulian2016}.
%The parameter tree-width is  more general and corresponds to a bigger class than the aforementioned classes and the fixed parameter tractability of $\GI$ parameterized by tree-width was a long standing open problem. Lokshtanov et al.~\cite{lokshtanov2017fixed} gave an $\fpt$ algorithm for this problem. 
%A more detailed list of $\fpt$ algorithms for $\GI$ in parameterized setting can be found in~\cite{Bulian2016}.

While there are many results on the parallel or the logspace complexity of problems in the parameterized framework~\cite{downey2013fundamentals}, very little is known in this direction for $\GI$.
The parameterized analogues of classical complexity classes have also been studied in~\cite{cai1997advice,flum2003describing,elberfeld2012space}.
The class $\PL$-$\mathcal{C}$ is the family of parameterized problems that are in $\mathcal{C}$ after a pre-computation on the parameter, where $\mathcal{C}$ is a complexity class.  In this paper we study  the graph isomorphism problem from a parameterized space and parallel complexity perspective.  
% Research on parameterized space complexity \cite{elberfeld2012space,cai1997advice} has been studied at the logspace analogues $\PL$-$\L$ of \fpt{}.
Recently Chandoo~\cite{chandoo2016deciding} showed that $\GI$ for circular-arc graphs is  in $\PL$-$\L$ when parameterized by the cardinality of an obstacle set.

%The parameterized complexity status of $\GI$  parameterized by clique-width remains a challenging open problem.

Since the graph isomorphism problem parameterized by tree-width has a logspace~\cite{elberfeld2016canonizing} as well as a separate \fpt{} algorithm~\cite{lokshtanov2017fixed} it is natural to ask if we can design a \emph{parameterized parallel} algorithm for this problem. In fact, the parallel complexity of $\GI$ parameterized by the well known but weaker parameter feedback vertex set number (FVS) is also unknown.   
%Kratsch et al.~\cite{kratsch2010isomorphism} gave an $\fpt$-algorithm for $\GI$ for another well known but weaker parameter, the feedback vertex set number. However, the parallel parameterized complexity of  $\GI$ with the feedback vertex set number as the parameter is not known. 
We make some progress in this direction by showing that $\GI$ parameterized by the size of a vertex cover, which is a weaker parameter than the FVS, is parallelizable in the parameterized setting.

Let $\mathcal{G}$ be a graph class characterized by a finite set of forbidden induced subgraphs (see Section~\ref{sec-3} for the formal definition). Kratsch et al.~\cite{kratsch2010isomorphism} gave an $\fpt$ algorithm for $\GI$ parameterized by the distance to $\mathcal{G}$ by taking a polynomial time colored graph isomorphism algorithm for graphs in $\mathcal{G}$ as a subroutine. In Section~\ref{sec-3}, we show that the result of~\cite{kratsch2010isomorphism} is parallelizable in the parameterized framework. More precisely, we give a $\PL$-$\AC^1$ algorithm for this problem. As a consequence, observe that $\GI$ parameterized by the distance to cographs is in $\PL$-$\AC^1$.  

Using bounded search tree method we also design a parallel recognition algorithm for graphs parameterized by the distance to $\mathcal{G}$. One would ask if the problem is in $\PL$-$\L$ using the same method as in \cite{cai1997advice} and~\cite{downey2013fundamentals}. However, the recent corrigendum  Cai et al.~\cite{Rod2017} suggests that this may need completely new ideas.

In the above mentioned parallel analogue of the result by  Kratsch et al.~\cite{kratsch2010isomorphism},   $\mathcal{G}$ is a class of graphs characterized by a finite set of forbidden induced subgraphs. Instead of that if we take $\mathcal{G}$ to be the set of bounded tree-width graphs then the parallel parameterized complexity is again open. Note that the analogous  preconditions of the theorem by  Kratsch et al.~\cite{kratsch2010isomorphism} in this scenario is met by the logspace $\GI$ algorithm for bounded tree-width graphs by Elberfeld et al. \cite{elberfeld2016canonizing}. In fact, the problem is open even when $\mathcal{G}$ is just the set of forests because this is the same problem: $\GI$ parameterized by feedback vertex set number. We study the graph isomorphism problem for \emph{bounded} distance to any graph class $\mathcal{G}$ under reasonable assumptions: the colored version of $\GI$ for the class  $\mathcal{G}$ and the recognition problem for $\mathcal{G}$ are in $\L$. We give a logspace isomorphism algorithm for such classes of graphs.  

In Section~\ref{sec-4}, we  show that \GI{} is in $\PL$-$\TC^0$ when parameterized by the vertex cover number. By using the recognition algorithm for graphs parameterized by the vertex cover number due to~\cite{bannach2015fast}, we first design a recognition algorithm for graphs parameterized by twin-cover number. We then prove that the graph isomorphism problem  parameterized by twin-cover  is in  $\PL$-$\TC^0$.

\begin{table}[]
\centering
\label{table1}
\begin{tabular}{c|c|c}
\toprule
\textbf{Parameter/Problem}             & \textbf{Recognition}         & \textbf{Graph Isomorphism}    \\ \midrule
Vertex Cover                          & $\PL$-$\AC^0$\cite{bannach2015fast}   &  $\PL$-$\TC^0$ [*] \\
Twin Cover                            & $\PL$-$\AC^0$    &  $\PL$-$\TC^0$ [*] \\
Distance to $\mathcal{H}$-free graphs & $\PL$-$\AC^{0\uparrow}$ [*]  &  $\PL$-$\AC^1$ [*] \\
Feedback Vertex set number            & Open for $\PL$-$\L$              & Open for  $\PL$-$\L$              \\
 \bottomrule
\end{tabular}
\caption{Parallel/Space complexity results/status on the graph isomorphism problem parameterized by various parameters. 
[*] indicates results presented in this paper.}
\end{table}
\section{Preliminaries}\label{sec-prelims}
The  basic definitions and notations of standard complexity classes are from~\cite{arora2009computational} and the definitions of parameterized versions of complexity classes are  from~\cite{elberfeld2012space,cai1997advice,stockhusen2017space}. 
A \emph{parameterized problem} is pair $(\mathcal{Q}, k)$ of a language $\mathcal{Q}\subseteq \Sigma^*$ and a parameterization $k:\Sigma^*\rightarrow \mathbb{N}$ that maps input instances to natural numbers, their parameter values\footnote{Often we write $k$ in stead of $k(x)$.}.
The class $\PL$-$\mathcal{C}$ is defined to be the family of problems that are in $\mathcal{C}$ after a precomputation on the parameter where $\mathcal{C}$ is a complexity class.
\begin{definition}~\cite{elberfeld2012space}
For a complexity class $\mathcal{C}$, a parameterized problem $(\mathcal{Q}, k)$ belongs to the \emph{para class} $\PL$-$\mathcal{C}$ if there is an alphabet $\Pi$, a computable function $\pi: \mathbb{N}\rightarrow \Pi^*$ and a language $A \subseteq \Sigma^*\times \Pi^*$ with $A \in \mathcal{C}$ such that for all $x\in \Sigma^*$
we have $x\in \mathcal{Q} \Leftrightarrow (x, \pi(k(x))) \in A$.
\end{definition}
If the complexity class $\mathcal{C}$ is $\L$ then we get the complexity class $\PL$-$\L$. The following equivalent definition of $\PL$-$\L$ is convenient when designing $\PL$-$\L$ algorithms.  
\begin{definition}~\cite{elberfeld2012space}
 A parameterized problem $(\mathcal{Q}, k)$ over $\varSigma$ is in $\PL$-$\L$ if there is function $f:\mathbb{N} \rightarrow \mathbb{N}$ such that the question $x \in \mathcal{Q}$ can be decided within space $f(k)+O(\log|x|)$.
\end{definition}
The parameterized parallel complexity classes are defined by using the basic complexity classes in place of $\L$ in above and basic gates ($AND$ and $OR$ gates) as follows~\cite{stockhusen2017space}:

\emph{$\PL$-$\AC^i$} (\emph{$\PL$-$\TC^i$}): The class of languages that are decidable via family of circuits over basic gates (resp. together with threshold gates) with unbounded fan-in, size $O(f(k)n^{O(1)})$, and depth $O(f(k)+\log^i n)$ if $i>0$ and depth $O(1)$ if $i=0$.
%\begin{itemize}
    %\item \emph{$\PL$-$\NL$}: The class of languages that are decidable via a non deterministic Turing machine that uses at most $O(f(k)+\log n)$ work space.
    %\item \emph{$\PL$-$\AC^i$}: The class of languages that are decidable via family of circuits over basic gates with unbounded fan-in, size $O(f(k)n^{O(1)})$, and depth $O(f(k)+\log^i n)$ if $i>0$ and depth $O(1)$ if $i=0$. 
     %\item \emph{$\PL$-$\TC^i$}: The class of languages that are decidable via family of circuits over basic gates together with threshold gates, with unbounded fan-in, size $O(f(k)n^{O(1)})$, and depth $O(f(k)+\log^i n)$ if $i>0$ and depth $O(1)$ if $i=0$. 
  %   \item \emph{$\PL$-$\NC^i$}: The class of languages that are decidable via family of circuits over basic gates with bounded fan-in, size $O(f(k)n^{O(1)})$, and depth $O(f(k)+\log^i n)$ if $i>0$ and depth $O(1)$ if $i=0$. 
%\end{itemize}
From the definition of $\PL$-$\mathcal{C}$, we know that for two complexity classes $\mathcal{C}$ and $\mathcal{C}'$, $\mathcal{C}\subseteq \mathcal{C}'$ if and only if $\PL$-$\mathcal{C} \subseteq \PL$-$\mathcal{C}'$ \cite{bannach2015fast}. Hence we have the following relation between complexity classes $\PL$-$\AC^0 \subsetneq \PL$-$\TC^0  \subseteq \PL$-$\L \subseteq \PL$-$\AC^1$.
There exists a circuit class $\PL$-$\AC^{0\uparrow}$ in between $\PL$-$\AC^0$ and $\PL$-$\AC^1$ which is strictly more powerful than  $\PL$-$\AC^0$. The definition of $\PL$-$\AC^{0\uparrow}$ is as follows.
\begin{definition}~\cite{bannach2015fast}
$\PL$-$\AC^{0\uparrow}$ is a class of languages that are decidable via family of circuits over basic gates with unbounded fan-in, size $O(f(k)n^{O(1)})$, and depth $g(k)$ where $f$ and $g$ are computable functions.
\end{definition}
The depth of the circuits in this class is bounded by a  function that depends only on the parameter. We have $\PL$-$\AC^0 \subsetneq \PL$-$\AC^{0\uparrow} \subseteq \PL$-$\AC^1 $ and $\PL$-$\AC^0 \subsetneq \PL$-$\L \subseteq \PL$-$\AC^1 $. We do not know relation between $\PL$-$\AC^{0\uparrow}$ and $\PL$-$\L$. 
The computational versions of all the above circuit classes can be defined in the usual manner by having multiple output gates.

In this paper, the graphs we consider are undirected and simple. For a graph $G =(V,E)$ , let $V(G)$ and $E(G)$ denote the vertex set and edge set of $G$ respectively. An edge $ \{u,v\} \in E(G)$ is denoted as $uv$ for simplicity. 
For a  subset $S \subseteq V(G)$, the graph $G[S]$ denotes the subgraph of $G$ induced by the vertices of $S$. We use notation $G \setminus S$ to refer the graph obtained from $G$ after removing the vertex set $S$. For a vertex $u\in V(G)$,
$N(u)$ denotes the set of vertices adjacent to $u$ and $N[u]=N(u) \cup \{u\}$. For a set $X\subseteq V(G)$, $N(X)$ denoted as $\cup_{v\in X}N(v)$.

In this paper we study problems similar to the graph modification problems where given a graph $G$, and a graph class $\mathcal{G}$ the task is to apply some graph operations (such as vertex or edge deletions) on $G$ to get a graph in $\mathcal{G}$. For example, if $\mathcal{G}$ is the class of edgeless graphs then the number of vertices to be deleted from graph $G$ to make it edgeless is the vertex cover problem. 
For a graph class $\mathcal{G}$, the \emph{distance to $\mathcal{G}$} of a graph $G$ is the minimum number of vertices to be deleted from $G$ to get a graph in $\mathcal{G}$. For a positive integer $k$, we use $\mathcal{G}+kv$ to denote the family of graphs such that each graph in this family can be made into a graph in $\mathcal{G}$ by removing at most $k$ vertices. 

\emph{Cographs} are $P_4$-free graphs i.e., they do not contain any induced paths on four vertices. \emph{Interval graphs} are the intersection graphs of a family of intervals on the real line. A graph is a \emph{threshold graph} if it can be constructed recursively by adding an isolated vertex or a universal vertex.

The parameterized \vc{} problem has input a graph $G$ and a positive integer $k$. The problem is to decide the existence of a vertex set $X\subseteq V(G)$ of size at most $k$ such that for every edge $uv \in E(G)$, either $u\in X$ or $v \in X$. A \emph{minimal vertex cover} of a graph is a vertex cover that does not contain another vertex cover.

%{**pending Vertex cover definition**}
%\problembox{Vertex Cover}{A graph $G=(V,E)$ and a positive integer $k$}{$k$}{Is there a set of vertices $V'\subseteq V$ of size at most $k$ such that for every edge $uv \in E$, either $u\in V'$ or $v \in V'$?}

\begin{definition}
Let $G$ be a graph. The set $X\subseteq V(G)$ is said to be \emph{twin-cover} of $G$ if for every edge $uv \in E(G)$ either \\
(a) $u \in X$ or $v \in X$, or  (b) $u$ and $v$ are twins\footnote{Two vertices  $u$ and $v$ are \emph{twins} if $N(u)\setminus \{v\}=N(v)\setminus \{u\}$.}.
\end{definition}
An edge between a pair of twins is called a \emph{twin edge}. 
The graph $G'$ obtained by removing a twin-cover from $G$ is a disjoint collection of cliques \cite{ganian2015improving}.

A \emph{kernel} for a parameterized problem $\mathcal{Q}$ is an algorithm which transforms an instance $(I,k)$ of $\mathcal{Q}$ to an equivalent instance $(I',k')$ in polynomial time such that $k' \leq k$ and $|I'| \leq f(k)$  for some computable function $f$. For more details on parameterized complexity see~\cite{downey2013fundamentals}.

In this paper, a coloring of a graph is just a mapping of the vertices of a graph to a set of colors, and it need not be proper.
\begin{definition}
The \emph{colored graph isomorphism problem} is to decide the existence of a color preserving isomorphism between a pair of colored graphs $G=(V,E)$ and $G'=(V',E')$, i.e., there exists a bijection mapping $\varphi:V\rightarrow V'$, satisfying the following conditions: 1) $(u,v)\in E \Leftrightarrow (\varphi(u),\varphi(v))\in E'$ for all $u,v \in V$ 2) $color(v)=color(\varphi(v))$ for all $v \in V$.
\end{definition}

\section{\GI{} for Distance to a Graph Class is in $\PL$-$\AC^{1}$} \label{sec-3}
In this Section, first we give a generic method to solve \GI{} for graphs from $\mathcal{G}+kv$ in $\PL$-$\L$ provided there is a logspace colored \GI{} algorithm for graphs in $\mathcal{G}$ and a $\PL$-$\L$ algorithm for enumerating vertex deletion sets.  
%\begin{theorem}\label{Th3}
%Let $\mathcal{G}$ be a class of graphs characterized by finitely many forbidden induced subgraphs $\mathcal{H}=\{H_1,H_2,\cdots,H_l\}$ with $|V(H_i)|\leq d$ for all $1\leq i \leq l$ where $d$ is a constant. Suppose finding all vertex deletion sets of $\mathcal{G}+kv$ is in  $\PL$-$\L$ and the colored graph isomorphism problem for graphs from  $\mathcal{G}$ is in $\L$.  Then the graph isomorphism problem for graphs from $\mathcal{G}+kv$ is in $\PL$-$\L$.
%\end{theorem}
\begin{theorem}\label{Th3}
Let $\mathcal{G}$ be a any graph class. Suppose enumerating all the vertex deletion sets of $\mathcal{G}+kv$ is in  $\PL$-$\L$ and the colored graph isomorphism problem for graphs from  $\mathcal{G}$ is in $\L$.  Then the graph isomorphism problem for graphs from $\mathcal{G}+kv$ is in $\PL$-$\L$.
\end{theorem}
\begin{proof}
Let $\mathcal{A_I}$ be a logspace algorithm to check whether two given input colored graphs $G_1$ and $G_2$ from $\mathcal{G}$ are isomorphic.
%Let $\mathcal{A_R}$ be a parameterized logspace algorithm to find all vertex deletion sets of any graph from $\mathcal{G}+kv$.
%Given a logspace algorithm $\mathcal{A}$ to check whether two given input colored graphs $G_1$ and $G_2$ from $\mathcal{G}$ are isomorphic or not.
 We assume that graphs $G_1$ and $G_2$ are at a distance at most $k$ from $\mathcal{G}$. If $G_1$ and $G_2$ belong to $\mathcal{G}$ then use the algorithm $\mathcal{A_I}$ to check the isomorphism between $G_1$ and $G_2$.
 Otherwise we consider a vertex deletion set $S \subseteq V(G_1)$ of minimum size (say $s$) such that $G_1 \setminus S \in \mathcal{G}$ and all possible vertex deletion sets $S_1,S_2,\cdots,S_m$ of size $s$ for $G_2$ such that $G_2 \setminus S_i \in \mathcal{G}$ for all $i \in [m]$ given as input. Notice that $m\leq f(k)$. 
 
 %It means find the $S$ and $S_i$ such that $G_1\setminus S \in \mathcal{G}$ and $G_2\setminus S_i \in \mathcal{G}$ where $1\leq i \leq m$.
 %The idea is to first find the mapping between one deletion set of $G_1$ to deletion set of $G_2$ by trying all possible choices of mapping from set $S$ to sets $S_1,S_2,\cdots,S_m$.  For each isomorphism of $S$ to $S_i$ we check the colored isomorphism between $G_1\setminus S$ to $G_2\setminus S_i$ by using the Algorithm $\mathcal{A}$.
For each $i \in [m]$, the algorithm iterates over all possible isomorphisms between $G[S]$ to $G[S_i]$, and tries to extend them to isomorphisms from $G_1$ to $G_2$ with the help of   the colored graph isomorphism algorithm $\mathcal{A_I}$ applied on some colored versions of $G_1 \setminus S$ and $G_2 \setminus S_i$, where the colors of the vertices are determined by their neighbors in the corresponding deletion set. A crucial observation is that any bijective mapping from $S$ to $S_i$ can be viewed as a string in $[s]^s$ and can be encoded as a string of length $O(k\log k)$. The string is $x_1\cdots x_s$ encodes the map that sends the $i$th vertex in $S$ to the $x_i$th vertex in $S_i$.  
 
 %We view the computation as the composition of two parameterized logspace Turing machines. The first one writes on the output tape a deletion set $S$ of minimum size $s$ followed by $V(G_1)\setminus S$ by using the algorithm $\mathcal{A_R}$. Similarly for $G_2$, it writes on the output all deletion sets $S_1,S_2,\cdots,S_m$ of size $s$ followed by $V(G_2)\setminus S_i$ for all $i$. The second Turing machine takes input $G_1$, $G_2$ along with $S$ followed by $V(G_1)\setminus S$ and $S_i$ followed by $V(G_2)\setminus S_i$ for all $i$ respectively. 
 
  For all $i$ algorithm iterates over all $s!$ bijective mappings from $S$ to $S_i$ using string of length $O(s\log s)$. Next it checks whether the bijective mapping is actually an isomorphism from $G_1[S]$ to $G_2[S_i]$. For each isomorphism $\varphi$ from $S$ to $S_i$, we need to check whether this isomorphism can be extended to an isomorphism from $G_1\setminus S$ to $G_2\setminus S_i$ by using algorithm $\mathcal{A_I}$. We color the vertices of $G_1\setminus S$ according to their  neighbourhood in $S$. Two vertices of $G_1\setminus S$ get same color if they have the same neighbourhood in $S$.
 A vertex $u$ in $G_1\setminus S$ and a vertex $v$ in $G_2\setminus S_i$ will get same color if $\varphi(N(u)\cap S)=N(v)\cap S_i$. 
 %Similarly we color the vertices of $G_2\setminus S_i$ according to the neighbourhood in $S_i$. 
 We query algorithm $\mathcal{A_I}$ with input the graphs $G_1\setminus S$ and $G_2\setminus S_i$ colored as above. If the algorithm $\mathcal{A_I}$ says `yes' then $G_1\cong G_2$ and the algorithm accepts the input. Otherwise it tries the next isomorphism from $S$ to $S_i$.  If for all $i$ and all isomorphisms from $S$ to $S_i$, the algorithm $\mathcal{A_I}$ rejects then the we conclude that  $G_1\ncong G_2$ and the algorithm rejects the input. 
 
We  note few more details of the algorithm to demonstrate that it  uses small space. The enumeration over the $S_i$'s can be done using a $\log m$ bit counter. To check if two vertices $u$ in $G_1\setminus S$ and $v$ in $G_2\setminus S_2$ have same color in logspace we can inspect each vertex in $G_1$, find out if it in $S$, find out if it is a neighbour of $u$, and check if its image under $\varphi$ is a neighbour of $v$. This needs constantly many counters.       \qed
%  We can find the deletion set $S$ and deletion sets $S_1,S_2,\cdots,S_m$ for $G_1$ and $G_2$ respectively in  by using the $\PL$-$\L$ algorithm $\mathcal{A_R}$.
%  To check the isomorphism between $S$ to $S_i$ by using relative indices of the input order and it will take at most $O(\log mk)$ space and some counters. We can enumerate all $k!$ isomorphisms between $S$ to $S_i$ by using one counter with $O(k \log k)$ space. Extend to isomorphism between $G_1\setminus S$ to $G_2\setminus S_i$ by using bounded number ($O(m(k!))$) of queries to algorithm $\mathcal{A_I}$ for all $1 \leq i \leq m$. Overall the algorithm can be implemented in $f(k)+O(\log n)$ space.
\end{proof}

%Next we turn our attention to the parallel analogue of the theorem by Kartsch et al. \cite{}.

Next we give a $\PL$-$\AC^{0\uparrow}$ recognition algorithm for graphs parameterized by the distance to a graph class $\mathcal{G}$ by using the bounded search tree technique, where $\mathcal{G}$ is characterized by finitely many forbidden induced subgraphs. 
%From this meta theorem we show that \GI{} parameterized by distance to cographs is in $\PL$-$\AC^1$.
\begin{definition}~\cite{kratsch2010isomorphism}\label{def1}\label{def6}
 A class $\mathcal{G}$ of graphs is \emph{characterized by finitely many forbidden induced subgraphs} if there is a finite set of graphs 
 $\mathcal{H}=\{H_1,H_2,\cdots,H_l\}$ 
 %with $|V(H_i)|\leq d$ for all $1\leq i \leq l$
 such that a graph $G$ is in $\mathcal{G}$ if and only if $G$ does not contain $H_i$ as an induced subgraph for any $i \in \{1,2,\cdots,l\}$.
 \end{definition}
Let $\mathcal{G}$ and $\mathcal{H}$ be classes as defined above.
%a class of graphs characterized by finitely many forbidden induced subgraphs $\mathcal{H}=\{H_1,H_2,\cdots,H_l\}$. 
We use the bounded search tree technique ~\cite{downey2013fundamentals,cai1996fixed} to find a set $S$ of size at most $k$ such that 
$G\setminus S \in \mathcal{G}$. In this method we can compute all deletion sets of size at most $k$. Let $d$ be the size of the largest forbidden induced subgraph in $\mathcal{H}$. The algorithm constructs a tree $T$ as follows. The root of the tree is labelled with the empty set. 
It finds a forbidden induced subgraph $H_i \in \mathcal{H}$ of size at most $d$ in $G$. Any vertex deletion set $S$ must contain a vertex of $H_i$.
We add $|V(H_i)|$ many children to the root labelled with vertices of $H_i$.
In general if a node is labelled with a set $P$, then we find a forbidden induced subgraph $H_j$ in $G \setminus P$ and create $|V(H_j)|$ many children for the node labeled $P$
and label each child with $P \cup \{v_i\}$, where $v_i \in H_j$. If there exists a node labeled with a set $S$ in $T$ of size at most $k$ such that 
$G \setminus S \in \mathcal{G}$, then $S$ is a required vertex deletion set. From this, we also know that there are at most $d^k$ minimal vertex deletion sets of size at most $k$. Using the same process we can also find all the minimal vertex deletion sets of size at most $k$. 

Cai et al.~\cite{cai1997advice} implemented bounded search tree method and kernelization to find the vertex cover in $\PL$-$\L$ in 1997. However, the implementation of bounded search tree method in $\PL$-$\L$ was reported to have some errors~\cite{Rod2017}. Thus, this paper seems to give the first implementation of bounded search tree method in $\PL$-$\AC^{0\uparrow}$. 
%We have $\PL$-$\AC^0 \subsetneq \PL$-$\AC^{0\uparrow} \subseteq \PL$-$\AC^1 $ and $\PL$-$\AC^0 \subsetneq \PL$-$\L \subseteq \PL$-$\AC^1 $.  
Let us recall form Section 2, that there is no known  relation between $\PL$-$\AC^{0\uparrow}$ and $\PL$-$\L$.
\begin{lemma}\label{RL4}
Let $\mathcal{G}$ be a class of graphs characterized by finitely many forbidden induced subgraphs $\mathcal{H}=\{H_1,H_2,\cdots,H_l\}$ with $|V(H_i)|\leq d$ for all $1\leq i \leq l$ where $d$ is a constant. On input a graph $G$, the problem of computing all vertex deletion sets of size at most $k$ is in $\PL$-$AC^{0\uparrow}$ where $k$ is the parameter.
\end{lemma}
\begin{proof}
The idea to implement the bounded search tree method in $\PL$-$\AC^{0\uparrow}$ is as follows: 
%Let $\mathcal{G}$ be a class of graphs characterized by finitely many forbidden induced subgraphs $\mathcal{H}=\{H_1,H_2,\cdots,H_l\}$.

Consider the set of all subsets of size at most $d$ that induce a forbidden subgraph in $G$. We order these subsets lexicographically to obtain a list $\mathcal{L}=A_1,\cdots,A_m$ where for each $i$, $G[A_i]$ is isomorphic to some graph in $\mathcal{H}$.  Notice that $m=O(n^d)$. The list $\mathcal{L}$ can be computed in $\PL$-$\AC^{0\uparrow}$ by first producing all subsets of $V(G)$ of size at most $d$ and then keeping only those that induce a subgraphs isomorphic to some $H$ in $\mathcal{H}$. Observe that any vertex deletion set must contain at least one vertex from each $A_i$ for all $i$.  The algorithm uses all strings $\Gamma=\gamma_1 \cdots \gamma_k\in [d]^k$ in parallel to pick the vertex deletion sets $S$ of size at most $k$ as follows: Let us concentrate on the part of the circuit that processes a particular string $\Gamma=\gamma_1 \cdots \gamma_k$. Initially the deletion set $S$ is empty. The algorithm puts the $\gamma_i$th vertex (in lexicographic order) of $A_1$ in $S$ if $|A_1|\geq \gamma_1$. If $|A_1|<\gamma_1$ the computation ends in this part of the circuit. Suppose the algorithm has already picked $i$ vertices using $\gamma_1\cdots \gamma_i$. It picks the $(i+1)$th vertex using $\gamma_{i+1}$. To do so it picks the first set $A_j$ in the list $\mathcal{L}$ such that $A_j\cap S=\phi$ (if $A_j\cap S\neq \phi$  we say that $A_j$ is `hit' by $S$). Then it puts the $\gamma_{i+1}$th vertex of $A_j$ in $S$ if $|A_j|\geq \gamma_{i+1}$. Otherwise the computation ends in the part processing $\Gamma$. If on or before reaching $\gamma_k$ we have obtained a set $S$ such that $A_j\cap S\neq \phi$ for all $j$, the algorithm has successfully found a vertex deletion set. We say that the algorithm is in \emph{phase $i$} if it processing $\gamma_i$.  

To see that the algorithm can be implemented in $\PL$-$\AC^{0\uparrow}$, we just need to observe that in each phase the algorithm has to maintain the list of sets in $\mathcal{L}$ that are not yet hit by $S$. The depth of the circuit is $O(k)$ and the total size is $d^k poly(n)$. \qed
\end{proof}
We implemented the bounded search tree method in $\PL$-$\AC^{0\uparrow}$. This implementation can be used not only to recognize the graph class defined in the Definition~\ref{def6}  but also, as we can show, for designing $\PL$-$\AC^{0\uparrow}$ algorithms for the problems \textsc{Restricted Alternating Hitting Set} and  \textsc{Weight $\leq k$ $q$-Cnf Satisfiability}.  
%We can show that some of the problems studied in ~\cite{downey2013fundamentals,cai1997advice} are in $\PL$-$\AC^{0\uparrow}$ by using the bounded search tree method. 
The problems are as follows:\\
{{\bf Problem 1:}~\cite{downey2013fundamentals,cai1997advice} \textsc{Restricted Alternating Hitting Set}}\\
{\bf Instance: } A collection $C$ of subsets of a set $B$ with $|S|\leq k_1$ for all $S \in C$.\\
{\bf Parameter:} Two positive integers $(k_1, k_2)$.\\
{\bf Question:} Does Player I have a win in at most $k_2$ moves in the following game? Players play alternatively and choose unchosen elements, until, for each $S \in C$ some member of $S$ has been chosen. The player whose choice this happens to be wins. \\
{{\bf Problem 2:}~\cite{downey2013fundamentals,cai1997advice} \textsc{Weight $\leq k$ $q$-Cnf Satisfiability}}\\
{\bf Instance: } Boolean formula $\varphi$ in conjunctive normal form with maximum clause size $q$ where $q$ is fixed. \\
{\bf Parameter:} A positive integer $k$.\\
{\bf Question:} Does $\varphi$ have a satisfying assignment with at most $k$ literals true?
%{{\bf Problem 3:} PLANAR DOMINATING SET}\\
%{\bf Instance: } Given a planar graph $G$. \\
%{\bf Parameter:} $k$\\
%{\bf Question:} Does $G$ have a dominating set of size $k$? (A \emph{dominating set} is a set $X\subseteq V(G)$ where for each $u \in V(G)$ there is a $v \in X$ such that $uv \in E(G)$.)\\

\begin{theorem}
 The following problems are in $\PL$-$\AC^{0\uparrow}$:\\
 (i) \textsc{Restricted Alternating Hitting Set}.\\
 (ii) \textsc{Weight $\leq k$ $q$-Cnf Satisfiability}.
% (iii) PLANAR DOMINATING SET.\\
\end{theorem}
Downey et al.~\cite{downey2013fundamentals} gave $\fpt$ algorithms for these two problems by using bounded search tree method. We implemented the bounded search tree method in $\PL$-$\AC^{0\uparrow}$. Thus, these two problems are also in  $\PL$-$\AC^{0\uparrow}$.

The next theorem is obtained by replacing the complexity class $\PL$-$\L$ by $\PL$-$\AC^1$ in Theorem~\ref{Th3}.
%An equivalent theorem can be obtained if we replace the complexity class $\PL$-$\L$ by $\PL$-$\AC^1$. 
The proof of the theorem uses similar ideas and the implementation is easier. Moreover, because of Lemma~\ref{RL4} we do not have to assume the existence of an algorithm that outputs all the vertex deletion sets.

\begin{theorem}\label{Th31}
Let $\mathcal{G}$ be a class of graphs characterized by finitely many forbidden induced subgraphs $\mathcal{H}=\{H_1,H_2,\cdots,H_l\}$ with $|V(H_i)|\leq d$ for all $1\leq i \leq l$ where $d$ is a constant. Suppose the colored graph isomorphism problem for graphs from  $\mathcal{G}$ is in $\AC^1$.  Then the graph isomorphism problem for graphs from $\mathcal{G}+kv$ is in $\PL$-$\AC^1$.
\end{theorem}

\begin{corollary}
 The graph isomorphism problem parameterized by the  distance to cographs is in $\PL$-$\AC^1$.
\end{corollary}
\begin{proof}
%Proof for (1) and (2), colored graph isomorphism problem of  interval graphs is in $\L$ (see  \cite{kobler2011interval} Theorem 4.7). Threshold graphs and  trivially perfect graphs are sub class of interval graphs. 
%By using Lemma~\ref{RL4} we can find the all deletion sets of distance to threshold graphs and trivially perfect graphs. 
%We solve graph isomorphism problem for distance to threshold graphs and distance to trivially perfect graphs is in $\PL$-$\AC^1$ by using Theorem~\ref{Th3} and logspace algorithm for colored graph isomorphism problem of threshold graphs and trivially perfect graphs. 
Recall that cographs are graphs without any induced $P_4$.  
 The colored graph isomorphism for cographs was shown to be in $\L$ using logspace algorithm to find the modular decomposition~\cite{grussien2017capturing}. From this along with Theorem~\ref{Th31} and Lemma~\ref{RL4}, we deduce that the graph isomorphism problem for distance to cographs is in $\PL$-$\AC^1$.\qed 
%Proof for (4), We know that colored graph isomorphism for cluster (collection of cliques) is in $\L$. From this, Theorem~\ref{Th3} and Lemma~\ref{RL4}, graph isomorphism problem for distance to cluster is in $\PL$-$\AC^1$.  
%In next section we show that \GI{} is in $\PL$-$\TC^0$ for (5),(6) and (7).
\end{proof}
As a consequence of the above corollary, we can also solve graph isomorphism problem for some of the other graph classes e.g.,  distance to cluster (disjoint union of cliques), distance to threshold graphs in $\PL$-$\AC^1$ by using the generalized meta Theorem~\ref{Th31}.

%Graph isomorphism problem for some of the graph classes i.e Distance to split graphs and distance to line graphs not known to be in $\L$ because split graphs and line graphs are Graph isomorphism complete.
For larger parameters like vertex-cover, distance to clique and twin-cover, we can get better complexity theoretic results which we discuss in the following section.

\section{\GI{} Parameterized by Vertex Cover is in $\PL$-$\TC^0$ } \label{sec-4}
In this section we give a parameterized parallel algorithm for \GI{} parameterized by vertex cover. 
%and graphs with bounded twin-cover. 
Sam Buss~\cite{buss1993nondeterminism} showed that \vc{}  admits a polynomial kernel. Based on this kernelization result, Cai et al.~\cite{cai1997advice}, Elberfeld et al.~\cite{elberfeld2012space} and Bannach et al.~\cite{bannach2015fast} showed that \vc{} is in $\PL$-$\L$, $\PL$-$\TC^0$ and $\PL$-$\AC^0$ respectively. 
%proved that \vc{}  is in .   observed that \vc{} is in . Recently   in .
These methods not only determines the existence of a vertex cover of size at most $k$ but can also output all vertex covers of size at most $k$ in $\PL$-$\AC^0$. 
%But we need all minimum vertex covers of size at most $k$ to solve the graph isomorphism problem for graphs with bounded vertex cover. Cai et al. also given bounded search tree method to show \vc{} is in $\PL$-$\L$. Even in this algorithm outputs all minimum vertex covers of size at most $k$ with input of graph $G$ and parameter $k$. 
We give a brief overview of the procedure to enumerate all vertex covers of size at most $k$ by using kernelization method given in~\cite{cai1997advice,bannach2015fast,elberfeld2012space}.

Observe that any vertex of degree more than $k$ must belong to any vertex cover of a given graph $G$. For the graph $G=(V,E)$, consider the set $V_H=\{v\in V(G) \big| d(v)>k\}$. If $|V_H|$ is more than $k$ then we declare that there is no $k$ sized vertex cover. Let us assume $|V_H|=b$. Consider the set $V_L=\{v\in V(G)  \big|  d(v)\leq k \texttt{ and } N(v)\setminus V_H \neq \emptyset \}$ of vertices that have at least one neighbour outside $V_H$. Notice that none of the edges in $G[V_L]$ are covered by $V_H$. Let $S'$ be a vertex cover of $G[V_L]$. It is easy to see that $V_H\cup S'$ forms a vertex cover of $G$. On the other hand if $S$ is a vertex cover of $G$ then $V_L\cap S$ is a vertex cover of $G[V_L]$. If the cardinality of $V_L$ is more than $(k-b)(k+1)$ then reject (because the graph induced by vertices $V_L$ with $k-b$ vertex cover and all vertices degree bounded by $k$ has no more than $(k-b)(k+1)$ vertices). So the cardinality of $V_L$ is not more than $(k-b)(k+1)$. We can use the best known vertex cover algorithm~\cite{chen2010improved} to find the $(k-b)$ vertex cover on the sub graph induced by vertices $V_L$. Elberfeld et al.~\cite{elberfeld2012space} pointed that the parallel steps of  this process are the following:
\begin{itemize}
    \item [i.] Checking whether the vertex belongs to $V_H$.
    \item [ii.] Checking whether $|V_H|$ at most $k$.
    \item [iii.] Checking whether $|V_L|$ at most $k(k+1)$.
    \item [iv.] Computing the induced subgraph $G[V_L]$ from $G$ .
\end{itemize}
The computation of above steps can be implemented by $\PL$-$\AC^0$ circuits~\cite{bannach2015fast}.  The above process finds all vertex covers of size at most $k$ by enumerating all the $2^{|V_L|}$ possible binary strings on length $|V_L|$. 

\begin{theorem}\label{th2}
The graph isomorphism problem parameterized by vertex cover is in $\PL$-$\TC^0$. \end{theorem}
\begin{proof}
Given two input graphs $G$ and $H$ with vertex cover of size at most $k$, we need to test if  $G$ to $H$ are isomorphism in $\PL$-$\TC^0$. Using the kernelization method of Bannach et al.~\cite{bannach2015fast} we can recognize whether these two graphs have same sized vertex covers or not. For the graph $G$ we find a  minimal vertex cover $S$ of size at most $k$ and for graph $H$ we find all minimal vertex covers $S_1,S_2,\cdots,S_m$, each of size at most $k$. Notice that $m$ is at most $2^k$. 
%Consider we have first vertex cover set $S$ followed by $V(G_1)\setminus S$ for $G$ and first $S_1\cup S_2\cup \cdots \cup S_m$ followed by $V(G_1)\setminus S_1\cup S_2\cup \cdots \cup S_m$ for $H$ as input.
We know that if $G \cong H$ then $G[S]\cong H[S_i]$ for some $1\leq i \leq m$. We try all isomorphisms from the  minimal vertex cover $S$ of $G$ to each minimal vertex cover $S_i$ of $H$. Suppose $G[S]\cong H[S_i]$ via $\varphi$. We need to extend this isomorphism from the independent set $G\setminus S$ to $H\setminus S_i$. There are at most $k!$ isomorphisms between $G[S]$ to $H[S_i]$. The algorithm processes all pairs $(S,S_i)$ and all the isomorphisms in parallel. 

For each isomorphism $\varphi$, we need to check whether this $\varphi$ can be extended to an isomorphism between $G\setminus S$ to $H\setminus S_i$.
We partition the vertices of the graph $G \setminus S$ into at most $2^k$ sets (called `types') based on their neighborhood in $S$. For each $U\subseteq S$ let  $T_G(U,S)=\{u \in G \setminus S ~|~ N(u) = U\} $. It is not hard to see that $G\cong H$ if and only if there is a minimal vertex cover $S_i$ of $H$ and an isomorphism $\varphi$ from $G[S]$ to $H[S_i]$ such that for each $U \subseteq S$, $|T_G(U,S)|=|T_H({\varphi(U)},S_i)|$. The problem of testing whether $G$ is isomorphic to $H$ reduces to counting the number of vertices in each type. We represent each type using an $n$-length binary string, where $i^{th}$ entry is one if $v_i$ belongs to that type and zero otherwise. Since the \textsc{Bit Count}\footnote{Counting the number one's in $n$ length binary string.} problem is in $\TC^0$, counting the number of vertices in a type can be implemented using a  $\TC^0$ circuit. In summary, for each $S_i$ and each isomorphism between $G[S]$ and $H[S_i]$, and for each $U\subseteq S$ we check whether $|T_G(U,S)|=|T_H({\varphi(U)},S_i)|$.
This completes the proof. \qed

\end{proof}

\begin{corollary}
The graph isomorphism problem is in $\PL$-$\TC^0$ when parameterized by the distance to clique. 
\end{corollary}
\begin{proof}
We apply Theorem~\ref{th2} to the complements of the input graphs. \qed
\end{proof}
\begin{corollary}
The graph isomorphism problem parameterized by the size of a twin-cover is in $\PL$-$\TC^0$. \end{corollary}
\begin{proof}
To find the twin-cover, we first remove all the twin edges and then compute a vertex cover of size at most $k$ in the resulting graph as was done in \cite{ganian2015improving}. The first step runs through all edges and deletes an edge if it is a twin edge. Next it finds a vertex cover in the resulting graph which can be done in $\PL$-$\AC^0$ \cite{bannach2015fast}. Thus, computing all the  twin-covers can be done in $\PL$-$\AC^0$.

Now we describe the process of testing isomorphism. 
The idea for testing isomorphism of the input graphs parameterized by the size of a twin-cover is similar to that in the proof of   Theorem~\ref{th2}.  Let $S_1$ be a fixed twin-cover in $G_1$ and $S_2$ be a twin-cover in $G_2$ of same size. The algorithm processes all such $(S_1,S_2)$ pairs in parallel. First fix an isomorphism (say $\sigma$) from  $S_1$ to $S_2$ and try to extend it to $G_1\setminus S_1$ to $G_2\setminus S_2$. Again, all such isomorphisms are processed in parallel. 
We know that the graph $G\setminus S$ obtained by removing a twin-cover $S$ from $G$ is a disjoint collection of cliques. Any two vertices in a clique $C$ have same neighbourhood in $G$ i.e., if $u, v \in C$ then $N[u]=N[v]$. Thus, the `type' of a clique is completely determined by the neighbourhood of any of the vertices in the vertex deletion set, and the size of the clique. Formally, with respect to the isomorphism $\sigma$, a  clique $C_{g_1}$ in $G_1\setminus S_1$ and a clique $C_{g_2}$ in $G_2\setminus S_2$ have same \emph{type} if 1) $|V(C_{g_1})|=|V(C_{g_2})|$ and 2) $\sigma(N(C_{g_1}))=N(C_{g_2})$. The algorithm needs to check that the number of cliques in each type is same in both the graphs. This problem can again be reduced to instances of the \textsc{Bit Count} problem.

%For every clique $C_{g_1}$ in $G_1\setminus S_1$ the algorithm finds an isomorphic copy  $C_{g_2}$ in $G_2\setminus S_2$ by testing  1) $|V(C_{g_1})|=|V(C_{g_2})|$ and 2) $\sigma(N(C_{g_1}))=N(C_{g_2})$. The algorithm also counts the number $c$ of cliques isomorphically matching with $C_{g_1}$ in $G_1\setminus S_1$ itself. Then $C_{g_2}$ is the $c^{th}$ copy and it has to matched with  $c$ different cliques in $G_2\setminus S_2$. If we discover any of the above condition failing the then we can try the next clique in $G_2\setminus S_2$. This can be achieved  by trying all pairs of cliques in parallel. If the algorithm  does not find the required number of isomorphic copies for any one of the cliques in $G_1\setminus S_1$ then it tries the other  isomorphisms from $G_1[S_1]$ to $G_2[S_2]$ in parallel.
It is easy to see that, the above process can be implemented in $\PL$-$\TC^0$.\qed
\end{proof}
\section{Logspace \GI{} Algorithms for Bounded Distance to Graph Classes}
In this Section, we show that for fixed $k$ \GI{} for graphs in $\mathcal{G}+kv$ is in $\L$ if the colored \GI{} for graphs in $\mathcal{G}$ is in $\L$ where $\mathcal{G}$ is  a graph class. 
From this result we obtain that \GI{} for $cographs+kv$ and $interval+kv$ graphs is in $\L$. Note that these results are not in the parameterized complexity theory framework.  The proof of the following theorem given in appendix.
\begin{theorem}\label{Th1}
 Let $k$ be a fixed and  $\mathcal{G}$ be a class of graphs. Suppose the problem of deciding if a given graph is in $\mathcal{G}$ and the colored graph isomorphism problem for graphs in $\mathcal{G}$ is in $\L$. Then the graph isomorphism problem for graphs from $\mathcal{G}+kv$ is in $\L$.
 \end{theorem}
%\begin{definition}\cite{kratsch2010isomorphism}
 %A class $\mathcal{G}$ of graphs is \emph{characterized by finitely many forbidden induced subgraphs} if there is a finite set of graphs 
 %$\mathcal{H}=\{H_1,H_2,\cdots,H_l\}$ 
 %with $|V(H_i)|\leq d$ for all $1\leq i \leq l$
 %such that a graph $G$ is in $\mathcal{G}$ if and only if $G$ does not contain $H_i$ as an induced subgraph for any $i \in \{1,2,\cdots,l\}$.
%\end{definition}
 Suppose graph class $\mathcal{G}$ is define as in Definition~\ref{def1}.
 It is not hard to see the problem of deciding if a graph $G$ is in a class $\mathcal{G}$ characterized by finitely many forbidden induced subgraphs is in logspace (See Lemma~\ref{RL1} in the Appendix). The proof of the next corollary follows from Lemma~\ref{RL1} and Theorem~\ref{Th1}.

%For a class $\mathcal{G}$, we say that a graph $G$ has \emph{vertex deletion at most $k$ from} $\mathcal{G}$ if there is a deletion set $S$ of size at most $k$ vertices such that $G\setminus S \in \mathcal{G}$. 
\begin{corollary}\label{cr1}
 Let the graph class $\mathcal{G}$ be characterized by finitely many forbidden induced subgraphs $\mathcal{H}=\{H_1,H_2,\cdots,H_l\}$ with $|V(H_i)|\leq d$ for all $1\leq i \leq l$ where $d$ is a constant. The graph isomorphism problem for graphs with bounded vertex deletion from $\mathcal{G}$ is in $\L$ provided the colored graph isomorphism problem for graphs from  $\mathcal{G}$ is in $\L$. 
\end{corollary}

\begin{corollary}
 The graph isomorphism problem is in $\L$ for following graph classes: 1) distance to interval graphs 2) distance to cographs.
 %\begin{enumerate}
  %   \item distance to interval graphs
  %   \item distance to cographs
 %\end{enumerate}
\end{corollary}
\begin{proof}
The proof of (1), follows from  Theorem~\ref{Th1} and the  logspace algorithm for colored \GI{} for interval graphs (see~\cite{kobler2011interval}). \\
The proof of (2), follows from Corollary~\ref{cr1} and the logspace isomorphism algorithm for colored \GI{} for cographs~\cite{grussien2017capturing}.\qed
\end{proof} 
\section{Conclusion}
In this paper we showed that graph isomorphism problem is in $\PL$-$\TC^0$ when parameterized by the vertex cover number of the input graphs. 
We also studied the parameterized complexity of graph isomorphism problem for the class of graphs $\mathcal{G}$ characterized by finitely many forbidden induced subgraphs. We showed that
graph isomorphism problem is in $\PL$-$\AC^{1}$ for the graphs in  $\mathcal{G}+kv$ if  there is an $\AC^{1}$ algorithm for colored-$\GI$ for the graph class $\mathcal{G}$. 
From this result, we show that $\GI$ parameterized by the distance to cographs is in $\PL$-$\AC^1$. 

The following questions remain open.
Can we get a parameterized logspace algorithm for $\GI$ parameterized by feedback vertex set number? Does the  problem admit parameterized parallel algorithm?
Elberfeld et al.~\cite{elberfeld2016canonizing} showed that $\GI$ is in logspace for graphs of bounded tree-width. In this paper, we showed that $\GI{}$ for some subclasses of bounded clique-width graphs is in $\L$. It is an interesting open question to extend these results to bounded clique-width graphs.

\bibliographystyle{splncs03}
\bibliography{myrefs.bib}

\section{Appendix}
{\bf Proof of Theorem~\ref{Th1}}
\begin{proof}
The idea behind this proof is similar to that of Theorem~\ref{Th3}. Let $G_1$ and $G_2$ be the two input graphs.
The logspace graph isomorphism algorithm for graphs in $\mathcal{G}$ works via finding a vertex deletion set  $S_1$ for $G_1$ of size at most $k$. Next we iterate over all vertex deletion sets $S_2$ of the same size. The idea is to fix an isomorphism from $G_1[S_1]$ to $G_2[S_2]$ and check if the isomorphism can be extended to an isomorphism of the input graphs. To store the vertex deletion sets we need $O(k\log n)$ space in the work-tape.

We first describe how to find a vertex deletion set of a graph $G$. Choose a set $S$ of size at most $k$ vertices from $V(G)$ and test whether $G\setminus S \in \mathcal{G}$ by using the logspace algorithm (say $\mathcal{A}_r$) for deciding if an input graph is in $\mathcal{G}$. For every set $S$ of size at most $k$ from $V(G)$, if the recognition algorithm says $G\setminus S \notin \mathcal{G}$ then algorithm can infer that $G\notin \mathcal{G}+kv$. If for any of the sets, the algorithm $\mathcal{A}_r$ says $G\setminus S \in \mathcal{G}$ then algorithm outputs $S$ as vertex deletion set. The iteration of over sets of size at most $k$ can be easily implemented in logspace by using $k$ counters. Therefore, the whole process can be executed in logspace.

Let $S_1$ and $S_2$ be the vertex deletion sets of $G_1$ and $G_2$ obtained using the above logspace procedure. %The idea is to first find the mapping between one deletion set of $G_1$ to deletion set of $G_2$ by trying all possible choices of mapping from set $S_1$ to sets $S_1,S_2,\cdots,S_m$ where $m$ at most $n^k$.  
Fix a bijection $\sigma$ from $S_1$ to $S_2$ and test if it is an isomorphism form $G_1[S_1]$ to $G_2[S_2]$. If not we try the next $S_2$ in the lexicographic input order. Otherwise, we test if $\sigma$ can be extended to an isomorphism of the input graphs. The map $\sigma$ induces a coloring of the graphs $G_1'=G_1\setminus S_1$ and $G_2'=G_2\setminus S_2$. Two vertices in $G_1'$ ($G_2'$) get same color if they have the same neighbourhood in $S_1$ ($S_2$ respectively). A vertex $u$ in $G_1'$ and a vertex $v$ in $G_2'$ will get same color if $\sigma(N(u)\cap S_1)=N(v)\cap S_2$. It is easy to see that the resulting graphs have at most $2^k$ colors. Moreover, $\sigma$ can be extended to an isomorphism of $G_1$ and $G_2$ if and only if the colored versions of $G_1'$ and $G_2'$ are isomorphic.

Computing if two vertices $u$ and $u'$ in $G_1'$ have same color amounts to searching their neighbourhood in $S_1$. Since $S_1$ is in the work-tape this can be done in logspace. Similarly, by the fact that $S_2$ and $\sigma$ are in the work-tape, checking $\sigma(N(u)\cap S_1)=N(v)\cap S_2$ can also be performed in logspace. Since we can compute the colors in logspace, testing if colored $G_1'$ and $G_2'$ can be done in logspace using the logspace isomorphism test of colored graphs in $\mathcal{G}$. This completes the description of the algorithm.\qed
\end{proof}
 \begin{lemma}\label{RL1}
 Let $\mathcal{G}$ be a class of graphs characterized by finitely many forbidden induced subgraphs $\mathcal{H}=\{H_1,H_2,\cdots,H_l\}$ with $|V(H_i)|\leq d$ for all $1\leq i \leq l$ where $d$ is a constant. There is a logspace algorithm that on input a graph $G$ decides if $G \in \mathcal{G}$.
\end{lemma}
\begin{proof}
 Given a graph $G$, the aim is to check whether $G \in \mathcal{G}$. For this it is enough to check for each $i\in \{1,2,\cdots,l\}$ whether $G$ contains $H_i$ as an induced subgraph.
 The algorithm heavily uses the input order of the vertices of $G$. Let $d_i$ be the number of vertices in $H_i$. To check if $H_i$ appears as an induced subgraph of $G$, the algorithm picks vertices $v_1,v_2,\cdots,v_{d_i}$ from $V(G)$ and checks if these vertices induces $H_i$ in $G$. If not then the algorithm chooses a different set of $d$ vertices according to the input order of  $V(G)$ and repeats the same process. For each $i$, if none of the $d_i$-sized subsets of $V(G)$ forms $H_i$ then the algorithm concludes that $G \in \mathcal{G}$. Otherwise it concludes that $G\notin \mathcal{G}$. It is easy to see that this algorithm can be implemented in logspace because at each step we need $O(d \log n)$ space to store at most $d$ vertices of $G$ and constantly many counters.\qed
 \end{proof}
 
\end{document}